\documentclass[pra,twocolumn,floatfix,superscriptaddress,notitlepage,longbibliography]{revtex4-1}

\usepackage{graphicx, color, graphpap}
\usepackage{cancel}
\usepackage{enumitem}
\usepackage{amssymb}
\usepackage{amsthm}
\usepackage{multirow}
\usepackage[colorlinks=true,citecolor=blue,linkcolor=magenta]{hyperref}
\usepackage[T1]{fontenc}
\usepackage{bbm}
\usepackage{thmtools,thm-restate}
\usepackage{verbatim}
\usepackage{mathtools}
\usepackage{titlesec}
\usepackage{soul}

\long\def\ca#1\cb{} 

\newcommand{\abs}[2][]{#1| #2 #1|}

\newcommand{\ketbra}[2]{| \hspace{1pt} #1 \rangle \langle #2 \hspace{1pt} |}

\newcommand{\bramatket}[3]{\langle #1 \hspace{1pt} | #2 | \hspace{1pt} #3 \rangle}

\newcommand{\norm}[2][]{#1| \! #1| #2 #1| \! #1|}

\newcommand{\ket}[1]{|#1\rangle}               
\newcommand{\bra}[1]{\langle #1|}              
\newcommand{\dya}[1]{\ket{#1}\!\bra{#1}}

\newcommand{\poly}{\operatorname{poly}}
\newcommand{\chull}{\text{Conv}}

\newcommand{\rank}{\text{rank}}

\newcommand{\HC}{\mathcal{H}}
\newcommand{\IC}{\mathcal{I}}

\newcommand{\SC}{\mathcal{S}}

\newcommand{\Tr}{{\rm Tr}}

\renewcommand{\geq}{\geqslant}
\renewcommand{\leq}{\leqslant}

\newcommand{\ad}{^\dagger}

\newcommand*{\id}{\openone}

{}
{}

\newtheorem{lemma}{Lemma}
\newtheorem{proposition}{Proposition}

\newtheorem{definition}{Definition}

\begin{document}
\title{Generalized Measure of Quantum Fisher Information}

\author{Akira Sone}
\thanks{These authors contributed equally to this work.}
\affiliation{Aliro Technologies, Inc. Boston, Massachusetts 02135, USA}
\affiliation{Theoretical Division, Los Alamos National Laboratory, Los Alamos, New Mexico 87544, USA}
\affiliation{Center for Nonlinear Studies, Los Alamos National Laboratory, Los Alamos, New Mexico 87544, USA}
\affiliation{Quantum Science Center, Oak Ridge, Tennessee 37931, USA}

\author{M. Cerezo}
\thanks{These authors contributed equally to this work.}
\affiliation{Theoretical Division, Los Alamos National Laboratory, Los Alamos, New Mexico 87544, USA}
\affiliation{Center for Nonlinear Studies, Los Alamos National Laboratory, Los Alamos, New Mexico 87544, USA}
\affiliation{Quantum Science Center, Oak Ridge, Tennessee 37931, USA}

\author{Jacob L. Beckey}
\affiliation{Theoretical Division, Los Alamos National Laboratory, Los Alamos, New Mexico 87544, USA}
\affiliation{JILA, NIST and University of Colorado, Boulder, Colorado 80309, USA}
\affiliation{Department of Physics, University of Colorado, Boulder, Colorado 80309, USA}

\author{Patrick J. Coles}
\affiliation{Theoretical Division, Los Alamos National Laboratory, Los Alamos, New Mexico 87544, USA}
\affiliation{Quantum Science Center, Oak Ridge, Tennessee 37931, USA}

\begin{abstract} 
In this work, we present a lower bound on the quantum Fisher information (QFI) which is efficiently computable on near-term quantum devices. This bound itself is of interest, as we show that it satisfies the canonical criteria of a QFI measure. Specifically, it is essentially a QFI measure for subnormalized states, and hence it generalizes the standard QFI in this sense. Our bound employs the generalized fidelity applied to a truncated state, which is constructed via the $m$ largest eigenvalues and their corresponding eigenvectors of the probe quantum state $\rho_{\theta}$. Focusing on unitary families of exact states, we analyze the properties of our proposed lower bound, and demonstrate its utility for efficiently estimating the QFI.
\end{abstract}
\maketitle

\section{Introduction}
\label{sec:intro}

Quantum Fisher information (QFI)~\cite{Hayashi_2004,Jing20} quantifies the ultimate precision with which one can estimate a parameter $\theta$ from a $\theta$-dependent quantum state $\rho_{\theta}$ via the quantum Cram\'{e}r-Rao bound. This quantity is of fundamental importance for quantum metrology~\cite{giovannetti2011advances,demkowicz2012elusive,Degen16x,LucaReview,meyer2020variational}. Moreover, the QFI has been studied in the context of quantum phase transitions~\cite{Nori14,Ye16,macieszczak2016dynamical}, quantum information geometry~\cite{Amari07,Fujiwara95}, and quantum information~\cite{Pezze09,Modi11,Hyllus12, Kim18nonclassical, Sone18a,Sone19a,takeoka2016optimal,katariya2020geometric}. 

The general definition of the QFI is
\begin{align}
    I(\theta;\rho_{\theta})=\Tr[J_{\theta}^2\rho_{\theta}]\,,
\label{eq:standardQFI}
\end{align}
where $J_{\theta}$ is called symmetric logarithmic derivative (SLD) operator satisfying the following Lyapunov equation:
\begin{align}
    \partial_{\theta}\rho_{\theta}=\frac{1}{2}(J_{\theta}\rho_{\theta}+\rho_{\theta}J_{\theta})\,.
\label{eq:standardSLD}
\end{align}
Also, the QFI is associated with the standard fidelity between the exact state $\rho_{\theta}$ and the error state $\rho_{\theta+\delta}$ as
\begin{align}\label{eq:qfifid}
    I(\theta;\rho_{\theta})=8\lim_{\delta\to0}\frac{1-F(\rho_{\theta},\rho_{\theta+\delta})}{\delta^2}\,,
\end{align}
where $F(\rho_1,\rho_2)=|\!|\sqrt{\rho_1}\sqrt{\rho_2}|\!|_{1}$ is the standard fidelity, and with the trace norm given by $\norm{A}_1=\Tr[\sqrt{AA\ad}]$.

In spite of its theoretical significance, the QFI is in general a difficult quantity to compute. Calculating the SLD operator requires one to solve the Lyapunov equation, which in turn needs full knowledge of the exact state  $\rho_{\theta}$, which is not always known in practice. In addition, when employing Eq.~\eqref{eq:qfifid} to determine the QFI, one encounters the serious difficulty that there is no efficient algorithm to compute the fidelity between arbitrary states. The complexity of the classical algorithms for fidelity estimation can scale exponentially due to the exponentially large dimension of the density matrices with respect to the number of qubits~\cite{cerezo2020variationalfidelity}. But even quantum algorithms face complexity theoretic arguments \cite{watrous2002quantum}, and the fact that the nonlinear nature of fidelity implies that a finite number of copies of $\rho_\theta$ cannot lead to an exact computation of the fidelity. Hence, instead of exactly computing the QFI, one can {estimate the QFI by bounding it}~\cite{GirolamiPRL,GirolamiExpPRA}.

This is precisely the goal of this paper, where we introduce an efficiently computable lower bound for the QFI. Our bound is based on the truncated (and therefore subnormalized) state constructed by projecting the exact state $\rho_{\theta}$ into the subspace of its $m$-largest eigenvalues. Particularly, we focus on the family of quantum states of the form $\rho_{\theta}=W(\theta)\rho W\ad(\theta)$, where $\rho$ is called probe state, and we define $W(\theta)=e^{-i\theta G}$ with a Hermitian and $\theta$-independent generator $G$. As in Ref.~\cite{Mateo09}, we refer to the set of states of this form as a \textit{unitary family}. This family of states is general enough to describe phase estimation tasks, such as magnetometry~\cite{Taylor08, Dutta20, Degen16x}.

Our results are derived by employing the concepts of generalized fidelity~\cite{tomamichel2010duality,tomamichel2015quantum} and truncated states~\cite{cerezo2020variationalfidelity} to construct an efficiently computable quantity which we call truncated quantum Fisher information (TQFI). Our main results are a series of lemmas that prove that TQFI lower bounds the standard QFI, and that TQFI satisfies various properties, including most of the canonical criteria for a measure of QFI. In addition, we also introduce a quantity that we call the generalized Bures distance, from which we provide a geometrical interpretation to the TQFI. {We note that in our recent work~\cite{Beckey2020}, we have proposed a trainable variational quantum algorithm to estimate QFI and further prepare the optimal state for phase estimation by using TQFI.}

This paper is organized as follows. We first provide theoretical background in Sec.~\ref{sec:background}. Then, Sec.~\ref{sec:truncatedQFI} introduces the TQFI and its associated Hermitian SLD operator, and presents our main results. Finally, we offer some concluding remarks in Sec.~\ref{sec:conclusion}.

\section{Theoretical background}
\label{sec:background}

In this section we provide some theoretical background that will be useful to define the TQFI. Specifically, we discuss the generalized fidelity, a measure of distinguishability for subnormalized states. We then discuss how the generalized fidelity can be used to construct an upper bound for the standard fidelity. We remark that this bound will be the basis of the definition of the TQFI.

Let $\HC$ be a $d$-dimensional Hilbert space. A quantum state $\rho$ on $\HC$ is defined as a Hermitian, positive semi-definite operator of trace equal to 1. Hence, the set of normalized quantum states on $\HC$ can be defined as
\begin{equation}
    \SC_{=}(\HC)=\{\rho\,:\,\rho\ad=\rho\,,\,\rho\geq 0\,,\,\Tr[\rho]=1\}\,,
\end{equation}
which forms a convex set with real dimension $(d^2-1)$. Relaxing the normalization condition, one arrives at the following definition.
\begin{definition}[Subnormalized state]\label{def:subnorm}
A Hermitian, positive semi-definite operator $\tau$ on $\HC$ is said to be a subnormalized quantum state if $\Tr[\tau]\leq 1$.
\end{definition}

Definition~\ref{def:subnorm} allows us to introduce $\SC_{\leq}(\HC)$ as the set of subnormalized states on $\HC$, that is 
\begin{equation}
    \SC_{\leq}(\HC)=\{\tau\,:\tau\ad=\tau\,,\,\tau\geq 0\,,\,\Tr[\tau]\leq1\}\,.
\end{equation}
As schematically shown in Fig.~\ref{fig1}, it follows that $\SC_{=}(\HC)\subset\SC_{\leq}(\HC)$. Moreover, $\SC_{\leq}(\HC)$ has dimensionality $d^2$, and can be obtained as the convex hull of the set of quantum states and the zero operator $\SC_{\leq}(\HC)=\chull(0,\SC_{=}(\HC))$~\cite{cappellini2007subnormalized}. subnormalized quantum states have been used in quantum information theory as a convenient generalization of normalized quantum states  \cite{cappellini2007subnormalized,tomamichel2010duality,tomamichel2015quantum}. Moreover, exciting new work on near-term quantum algorithms utilizes truncated, and thus subnormalized, quantum states to avoid having to store an exponentially large density matrix, thus making the algorithms implementable on the noisy intermediate-scale quantum (NISQ) computers \cite{cerezo2020variationaleigensolver,cerezo2020variationalfidelity,Beckey2020}. This exciting new research direction is the primary motivation for this work.

In Refs.~\cite{tomamichel2010duality,tomamichel2015quantum} the authors introduced a generalization of the standard quantum fidelity to subnormalized states, which is known as the {\it generalized fidelity}, and which is given as follows.
\begin{definition}[Generalized fidelity]
Given two subnormalized states $\tau,\sigma\in\SC_{\leq}(\HC)$, the generalized fidelity between $\tau$ and $\sigma$ is
\begin{equation}
    F_*(\tau,\sigma)= \norm{\sqrt{\tau}\sqrt{\sigma}}_1+\sqrt{(1-\Tr[\tau])(1-\Tr[\sigma])}\,,
\end{equation}
where $\norm{A}_1=\Tr[\sqrt{AA\ad}]$ is the trace norm.
\end{definition}
Note that the generalized fidelity reduces to the standard fidelity $F$ if at least one of the two states is normalized. That is,
\begin{equation}
    F_*(\tau,\sigma)=F(\tau,\sigma)= \norm{\sqrt{\tau}\sqrt{\sigma}}_1\,,
\end{equation}
if $\tau$ or $\sigma$ is in $\SC_{=}(\HC)$.

\begin{figure}[t!]
    \centering
    \includegraphics[width=.8\columnwidth]{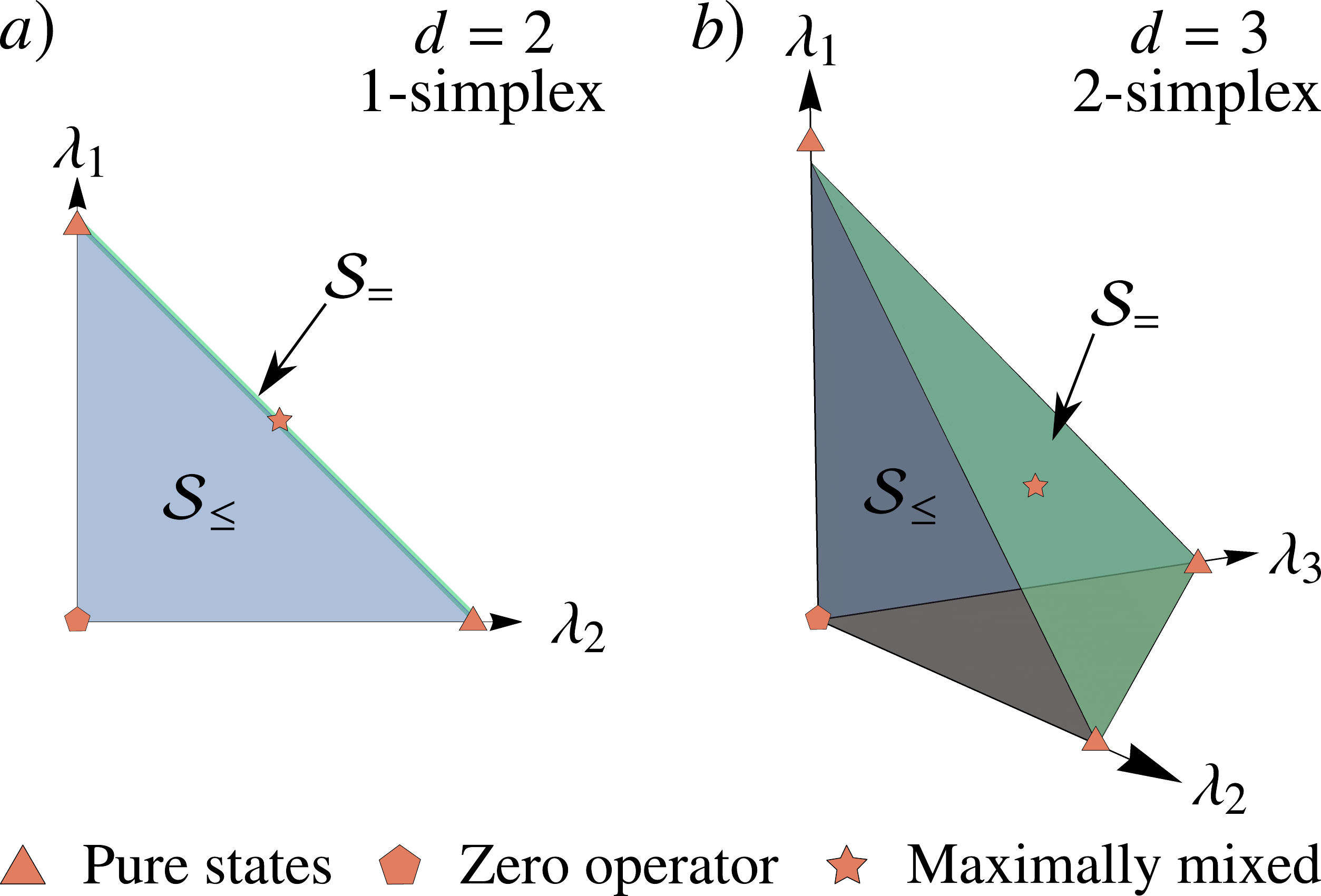}
    \caption{Set of eigenvalues $\{\lambda_i\}_{i=1}^d$ for a subnormalized quantum state. Simplexes are shown for a $d$-dimensional Hilbert space with (a) $d=2$ and (b) $d=3$. The eigenvalues of normalized states in $\SC_{=}(\HC)$ lie on a $(d^2-1)$-simplex: a line segment in (a), and a triangle in (b). The eigenvalues of a pure state lie on the edges of the simplex, while those of a mixed state are on the centroid of the simplex. The eigenvalues of subnormalized states in $\SC_{\leq}(\HC)$ can be obtained from the subnormalization condition $\sum_{i=1}^d\lambda_i\leq 1$. In the diagrams, the origin corresponds to the zero operator. 
    }
    \label{fig1}
\end{figure}

As shown in Refs.~\cite{tomamichel2010duality,tomamichel2015quantum,cerezo2020variationalfidelity}, the generalized fidelity has the following relevant properties:
\begin{itemize}
    \item Invariance under unitary transformations. Given two subnormalized states $\tau,\sigma\in\SC_{\leq}(\HC)$, and for any unitary $V$ in the unitary group $U(d)$ of degree $d$, we have 
    \begin{align}
    F_*(V\tau V\ad,V\sigma V\ad)=F_*(\tau,\sigma)\,.
    \end{align}
    \item Concavity. Given subnormalized states $\tau_1,\tau_2,\sigma_1,\sigma_2\in\SC_{\leq}(\HC)$, and a real number $q\in[0,1]$, then
    \begin{equation}
    \begin{split}
    F_*(q\tau_1+&(1-q)\tau_2,q\sigma_1+(1-q)\sigma_2)\\
    &\geq q F_*(\tau_1,\sigma_1)+(1-q)F_*(\tau_2,\sigma_2)\,.
    \end{split}
    \end{equation}
    \item Monotonicity under completely positive trace non-increasing (CPTNI) maps. Given two subnormalized states $\tau,\sigma\in\SC_{\leq}(\HC)$, and a CPTNI map $\Phi$, then 
    \begin{align}
    F_*(\tau,\sigma) \leq F_*(\Phi(\tau),\Phi(\sigma))\,.
    \end{align}
\end{itemize}

We note that CPTNI maps are the mathematical generalization of completely positive trace-preserving (CPTP) maps, which become useful when one allows for subnormalized quantum states \cite{tomamichel2015quantum,tomamichel2010duality,cappellini2007subnormalized}. Additionally, they have the physical interpretation of describing an experiment in which the measurement apparatus does not work with some probabilities~ \cite{cappellini2007subnormalized}. 

Let us now discuss how the generalized fidelity can be used to upper bound the standard fidelity. Consider a projection operator $\Pi$ which maps states to a subspace of $\HC$. Note that  $\Pi$ defines a CPTNI map as $\Phi(\rho)=\Pi\rho\Pi$ which maps states in $\SC_{=}(\HC)$ and in $\SC_{\leq}(\HC)$ to subnormalized states in $\SC_{\leq}(\HC)$.  Then, from the monotonicity under CPTNI maps of the generalized fidelity, the following bound on the standard fidelity $F(\rho,\widetilde{\rho})$ holds for any pair of normalized states $\rho$ and $\widetilde{\rho}$~\cite{cerezo2020variationalfidelity}:
    \begin{equation}\label{eq:upperbound}
        F(\rho,\widetilde{\rho})\leq F_*(\Pi\,\rho\,\Pi,\Pi\,\widetilde{\rho}\,\Pi)\,.
    \end{equation}

In Ref.~\cite{cerezo2020variationalfidelity}, the authors proposed an algorithm that can efficiently compute the upper bound in  Eq.~\eqref{eq:upperbound} for certain $\Pi$. Specifically, in that work, $\Pi$ is the operator that  projects onto the Hilbert space spanned by the eigenvectors of the $m$-largest eigenvalues of $\rho$.  That is, given the spectral decomposition $\rho=\sum_i\lambda_i\dya{\lambda_i}$, we define 
\begin{align}
    \Pi_\rho^m=\sum_{i=1}^m \dya{\lambda_i}\,.
\end{align}
This operator allows us to introduce the truncated states $ \rho^{(m)}$ and $ \widetilde{\rho}^{(m)}$:
\begin{equation}
\begin{split}
    \rho^{(m)}&=\Pi_\rho^m\rho \Pi_\rho^m=\sum_{i=1}^m\lambda_i\dya{\lambda_i}\,,\\
    \widetilde{\rho}^{(m)}&=\Pi_\rho^m\widetilde{\rho} \Pi_\rho^m\,,
\end{split}
\end{equation}
which leads to the following expression of the generalized fidelity for these states:
\begin{equation} \label{eq:gen-fidelity}
    F_*(\rho^{(m)}\!,\widetilde{\rho}^{(m)})= \Tr\left[\!\sqrt{T}\right]+\sqrt{(1\!-\!\Tr[\rho^{(m)}])(1\!-\!\Tr[\widetilde{\rho}^{(m)}])}\,.
\end{equation}
Here, $T$ is a positive semi-definite $m\times m$ operator given by
\begin{equation}
    T=\sum_{i,j=1}^m \sqrt{\lambda_i\lambda_j}\bramatket{\lambda_i}{\widetilde{\rho}}{\lambda_j}\ketbra{\lambda_i}{\lambda_j}\,.
\end{equation}
Finally, let us remark that the upper bound $F_*(\rho^{(m)}\!,\widetilde{\rho}^{(m)})\geq F(\rho,\widetilde{\rho})$ gets monotonically tighter with $m$, with equality holding if $m=\rank(\rho)$~\cite{cerezo2020variationalfidelity}.

\section{Truncated Quantum Fisher Information}
\label{sec:truncatedQFI}

From the discussions above, in this section, we introduce a generalized measure of the QFI definable with the subnormalized state, which we call TQFI, and show that it is a lower bound on the standard QFI. We then present some of its properties in the form of lemmas, which we prove in the Appendices, and present its geometrical interpretation in the space of subnormalized states.

 \subsection{Definition of the TQFI}
 
 Consider the (normalized) exact state $\rho_\theta$, and the (normalized) error state $\rho_{\theta+\delta}$. These states encode the information of an unknown parameter $\theta$ and of a shift $\delta$ in a probe state $\rho$ of rank $r$ as
 \begin{equation}
     \begin{split}
     \rho_\theta&=W(\theta)\rho W\ad(\theta)\,,\\ \rho_{\theta+\delta}&=W(\theta+\delta)\rho W\ad(\theta+\delta)\,,
 \end{split}
 \end{equation}
 with 
 \begin{equation}
    W(\theta)=e^{-i\theta G}\,,
 \end{equation}
where $G$ is a $\theta$-independent Hermitian operator.
Given the spectral decomposition of the exact state $\rho_\theta=\sum_i\lambda_i\dya{\lambda_i(\theta)}$ with  $\ket{\lambda_i(\theta)}=W(\theta)\ket{\lambda_i}$, we define the operator that projects onto the Hilbert space spanned by the eigenvectors corresponding to the $m$-largest eigenvalues of $\rho_\theta$ as $  \Pi_{\rho_\theta}^m=\sum_{i=1}^m \dya{\lambda_i(\theta)}$. Then, we define the truncated (subnormalized) states
\begin{align}\label{eq:truncation}
\begin{split}
    \rho^{(m)}_{\theta}&=\Pi_{\rho_\theta}^m\rho_\theta \Pi_{\rho_\theta}^m=\sum_{i=1}^m\lambda_i\dya{\lambda_i(\theta)}\,,\\ \rho^{(m)}_{\theta+\delta}&=\Pi_{\rho_\theta}^m\rho_{\theta+\delta} \Pi_{\rho_\theta}^m\,.
    \end{split}
\end{align}
Finally, we have the following definition for the TQFI.
\begin{definition}[Truncated quantum Fisher information]\label{def:TQFI}
Given an exact state $\rho_\theta$ and error state $\rho_{\theta+\delta}$ in $\SC_{=}(\HC)$,  let $\rho^{(m)}_{\theta}$ and $\rho^{(m)}_{\theta+\delta}$ denote their truncated versions according to~\eqref{eq:truncation} such that $\rho^{(m)}_{\theta},\rho^{(m)}_{\theta+\delta}\in\SC_{\leq}(\HC)$. The truncated quantum Fisher information is
\begin{equation}
        \mathcal{I}_*(\theta;\rho^{(m)}_{\theta})=8\lim_{\delta\to0}\frac{1-F_*(\rho^{(m)}_{\theta},\rho^{(m)}_{\theta+\delta})}{\delta^2}\,.
\label{eq:TQFIandGeneralizedFidelity}
\end{equation}
\end{definition}

 \subsection{TQFI as a lower bound}

From Eq.~\eqref{eq:upperbound} we have that the following lemma holds.

\begin{lemma}\label{lem:bound}
The TQFI of Definition~\ref{def:TQFI} is a lower bound for the QFI
\begin{equation}\label{eq:lowerbound}
      \mathcal{I}_*(\theta;\rho^{(m)}_{\theta})\leq   I(\theta;\rho_{\theta})\,,
\end{equation}
where $I(\theta;\rho_{\theta})$ is the QFI  defined in~\eqref{eq:qfifid}. 
In addition, the TQFI is monotonically increasing with $m$, i.e.,
\begin{equation}\label{eq:lowerbound2}
      \mathcal{I}_*(\theta;\rho^{(m)}_{\theta})\leq    \mathcal{I}_*(\theta;\rho^{(m+1)}_{\theta})\,,
\end{equation}
with the equality in~\eqref{eq:lowerbound} holding if $m=r$, where $r=\rank(\rho)$.
\end{lemma}

Lemma~\ref{lem:bound} provides an operational meaning of the TQFI as a lower bound on the standard QFI.  We remark that since the generalized fidelity is a tight bound for high purity states, the TQFI will also be a tight bound on the QFI in this case.

\subsection{Computation of TQFI}

Let us briefly discuss how the TQFI can be computed. We refer the reader to our work~\cite{Beckey2020} for a much more detailed description of the estimation of TQFI. As previously mentioned, the generalized fidelity can be efficiently computed for $m\in O(\poly(\log(d)))$ via a variational hybrid quantum-classical algorithm~\cite{cerezo2020VQAreview} called the variational quantum fidelity estimation algorithm in Ref.~\cite{cerezo2020variationalfidelity}, which uses state diagonalization as a subroutine~\cite{VQSD,cerezo2020variationaleigensolver}. Assuming this state {diagonalization} subroutine is efficient, it follows that one can efficiently approximate the TQFI and lower bound the QFI by using the algorithm in Ref.~\cite{cerezo2020variationalfidelity} and computing $\left(1-F_*(\rho^{(m)}_{\theta},\rho^{(m)}_{\theta+\delta})\right)/\delta^2$ for small $\delta$.

\subsection{Properties of the TQFI}

To further understand the meaning of the TQFI, it is useful to express this quantity in the representation of the eigenbasis of $\rho$.
\begin{lemma}
\label{lemma:TQFIexpression}
The TQFI of Definition~\ref{def:TQFI} can be written as 
\begin{align}
 \mathcal{I}_*(\theta;\rho^{(m)}_{\theta}) =&
    4\sum_{i,j=1}^m \lambda_i|G_{ij}|^2-8\sum_{i,j=1}^m\frac{\lambda_i\lambda_j}{\lambda_i+\lambda_j}|G_{ij}|^2
    \label{eq:TQFIexpression}
\end{align}
where $G_{ij}=\bramatket{\lambda_i}{G}{\lambda_j}$, and where we recall that $\lambda_i=0$ for $i>r$.
\end{lemma} 

Recalling that the standard QFI can be expanded in the eigenbasis of $\rho$ as~\cite{toth2014quantum}
\begin{equation}
\label{eq:QFIexpression}
I(\theta;\rho_{\theta}) =
    4\sum_{i,j=1}^d \lambda_i|G_{ij}|^2-8\sum_{i,j=1}^d\frac{\lambda_i\lambda_j}{\lambda_i+\lambda_j}|G_{ij}|^2
\end{equation}
with again $\lambda_{i,j}=0$ for $i,j>r$,  we can see that the first two terms in~\eqref{eq:TQFIexpression} are simply obtained by truncating the summations of~\eqref{eq:QFIexpression} so that $i,j=1,\cdots, m$; and while this may seem like the natural way to generalize the QFI to subnormalized states, the derivation of Eq. \eqref{eq:TQFIexpression} and the proofs of the properties required for it to satisfy the canonical criteria of a QFI measure are non-trivial. Before listing these important properties, let us consider an alternative definition for the TQFI by introducing the truncated symmetric logarithmic derivative (TSLD).
\begin{definition}[Truncated symmetric logarithmic derivative operator]\label{def:nTSLD}
Given a subnormalized truncated exact state $\rho^{(m)}_{\theta}\in\SC_{\leq}(\HC)$ defined according to~\eqref{eq:truncation}, the TQFI of Definition~\ref{def:TQFI} can be expressed as 
\begin{align}
    \mathcal{I}_*(\theta;\rho^{(m)}_{\theta})=\Tr\left[ L_{\theta}^2 \rho^{(m)}_{\theta}\right]\,,
\label{eq:truncatednSLDmain}
\end{align}
where 
\begin{align}\label{eq:nTSLD}
     L_{\theta}  = 2\sum_{i,j=1}^{m} \frac{\bra{\lambda_i(\theta)}\partial_{\theta}\rho_{\theta}^{(m)}\ket{\lambda_j(\theta)}}{\lambda_{i}+\lambda_{j}}\ket{\lambda_{i}(\theta)}\bra{\lambda_{j}(\theta)}\,
\end{align}
is the TSLD operator. For the unitary families, we particularly have
\begin{align}\label{eq:nTSLDunitary}
L_{\theta}  = 2i\sum_{i,j=1}^{m} \frac{\lambda_{i}-\lambda_{j}}{\lambda_{i}+\lambda_{j}}\bra{\lambda_{i}}G\ket{\lambda_{j}}\ket{\lambda_{i}(\theta)}\bra{\lambda_{j}(\theta)}\,.
\end{align}
\end{definition}
As we can see, the TSLD is simply obtained by truncating the summation of the SLD operator.

From the previous definitions and lemmas we can derive the following properties of the TQFI.
\begin{lemma}
\label{lemma:TQFIproperty}
From the definition of the TQFI,  for the unitary families $\rho_{\theta}=W(\theta)\rho W\ad(\theta)$, $\mathcal{I}_{*}(\theta,\rho^{(m)}_{\theta})$ satisfies the following properties:
\begin{itemize}
    \item Invariance under unitary transformations. Given a truncated subnormalized state $\rho^{(m)}_{\theta}\in \SC_{\leq}(\HC)$, for any $\theta$-independent unitary $V$ in the unitary group $U(d)$ of degree $d$, we have
    \begin{align}
    \mathcal{I}_{*}(\theta,V\rho^{(m)}_{\theta}V\ad)=\mathcal{I}_{*}(\theta,\rho^{(m)}_{\theta})\,.
    \end{align}
        \item Convexity. For two truncated subnormalized states $\rho^{(m)}_{\theta},\xi^{(m')}_{\theta}\in \SC_{\leq}(\HC)$ with $\rho^{(m)}_{\theta}=\Pi_{\rho_\theta}^m\rho_\theta \Pi_{\rho_\theta}^m$ and $\xi^{(m')}_{\theta}=\Pi_{\xi_\theta}^{m'}\xi_\theta \Pi_{\xi_\theta}^{m'}$, with a real number $q\in[0,1]$, we have 
    \begin{equation}
    \begin{split}
    \mathcal{I}_{*}(\theta;&q\rho^{(m)}_{\theta}+(1-q)\xi^{(m')}_{\theta}  )\\
    &\leq q \mathcal{I}_{*}(\theta;\rho^{(m)}_{\theta})+(1-q)\mathcal{I}_{*}(\theta;\xi^{(m')}_{\theta})\,.
    \end{split}
    \end{equation}
    \item Monotonicity under CPTNI maps. Given a truncated subnormalized state $\rho^{(m)}_{\theta}\in \SC_{\leq}(\HC)$,  and  a CPTNI map $\Phi$, we have 
    \begin{align}
    \mathcal{I}_{*}(\theta,\Phi(\rho^{(m)}_{\theta}))\leq \mathcal{I}_{*}(\theta,\rho^{(m)}_{\theta})\,.
    \end{align}
    \item Subadditivity for product of truncated states. Given a product of truncated states $\sigma=\bigotimes_{k}\rho^{(m_k)}_{k,\theta}$, where $\rho^{(m_k)}_{k,\theta}=\Pi_{\rho_{k,\theta}}^{m_k}\rho_{k,\theta} \Pi_{\rho_{k,\theta}}^{m_k}$, then we have 
    \begin{align}
    \mathcal{I}_{*}(\theta;\sigma)\leq\sum_{k}\mathcal{I}_{*}(\theta;\rho^{(m_k)}_{k,\theta})\,.
    \end{align}
    \item Additivity for direct sum of truncated states. Given a direct sum of truncated states $\sigma=\bigoplus_{k}\mu_{k}\rho^{(m)}_{k,\theta}$, where $\rho^{(m_k)}_{k,\theta}=\Pi_{\rho_{k,\theta}}^{m_k}\rho_{k,\theta} \Pi_{\rho_{k,\theta}}^{m_k}$, and where $\mu_k$ are $\theta$-independent coefficients such that $0<\sum_{k}\mu_{k}\leq 1$, we have 
    \begin{align}
    \mathcal{I}_{*}(\theta;\sigma ) =\sum_{k}\mu_{k}\mathcal{I}_{*}(\theta;\rho^{(m_k)}_{k,\theta})\,.
    \end{align}
\end{itemize}
\end{lemma}
Note that the TQFI satisfies the same properties as those that the standard QFI satisfies (see Ref.~\cite{Jing20} for a review of the properties of the QFI), except for the additivity for product of states. Here, the TQFI satisfies instead a subadditivity property which naturally follows from the fact that the states are subnormalized.

Let us finally discuss the geometric interpretation of the TQFI. From Eq.~\eqref{eq:TQFIandGeneralizedFidelity} we first define the \textit{generalized Bures distance}.  
\begin{definition}[Generalized Bures Distance]
\label{def:GBD}
Given two subnormalized states $\tau,\sigma\in\SC_{\leq}(\HC)$, the generalized Bures distance is
\begin{align}
        B_{*}^{2}(\tau,\sigma)=2(1-F_{*}(\tau,\sigma))\,.
\label{eq:GeneralizedBuresDistance}
\end{align}
\end{definition}
Here we remark that the generalized Bures distance is closely related to the purified distance for subnormalized states introduced in Refs.~\cite{tomamichel2010duality,tomamichel2015quantum}. Hence, the following lemma holds. 
\begin{lemma}
\label{lemma:Bures}
Given two subnormalized states $\tau,\sigma\in\SC_{\leq}(\HC)$, the generalized Bures distance $B_{*}^{2}(\tau,\sigma)$ is a distance metric on the space of subnormalized states. 
\end{lemma}
Then, for the truncated exact state$\rho^{(m)}_{\theta}$ defined in~\eqref{eq:truncation}, we can obtain the following result.
\begin{lemma}
\label{lemma:JandB}
Let $B_{*}^2 (\rho^{(m)}_{\theta},\rho^{(m)}_{\theta+\delta})$ be the generalized Bures distance. Then, for $\abs{\delta}\ll 1$, we have
\begin{align}
    B_{*}^2(\rho^{(m)}_{\theta},\rho^{(m)}_{\theta+\delta}) = \frac{1}{4}\mathcal{I}_{*}(\theta;\rho^{(m)}_{\theta})\delta^2+O(\delta^3)\,.
\end{align}
\end{lemma}
Lemma~\ref{lemma:JandB} provides a geometrical interpretation for the TQFI  as being related to the curvature of the generalized Bures distance in the space of subnormalized states. 

\bigskip

\section{Conclusion}
\label{sec:conclusion}
In conclusion, we have introduced the TQFI, which is demonstrated to be an efficiently computable lower bound on the quantum Fisher information. {This quantity can be used for estimating QFI and to prepare the optimal state for metrology via the variational quantum algorithms on the near-term quantum computers}. Specifically, the TQFI can be obtained from the generalized fidelity between the states obtained by projecting the exact state $\rho_\theta$ and error state $\rho_{\theta+\delta}$ onto the subspace spanned by the largest $m$ eigenvalues of $\rho_\theta$. For unitary families, we have proven that the TQFI satisfies the criteria of the quantum Fisher information for subnormalized states. In addition, we have revealed the geometrical interpretation of the TQFI by introducing a generalized Bures distance, a distance measure on subnormalized states. 

This lower bound can be employed to efficiently estimate the quantum Fisher information. This is especially useful in the context of quantum sensing, where one is interested in maximizing the quantum Fisher information. Hence, one can use our lower bound as a means to prepare states that maximize quantum Fisher information, to enhance sensing performance of the quantum sensors. Moreover, the quantum Fisher information is often used to witness metrologically useful entanglement~\cite{Hyllus12,Qin19}; therefore, an interesting future research direction will be exploring the use of TQFI for the entanglement witness in condensed matter systems.
\\
\\
\textit{Note.} {Recently, Yamagata ~\cite{Yamagata20} studied the properties of the quantum monotone metrics under the CPTNI maps. We note that TQFI belongs to this class.}

\section*{Acknowledgements}
This work was supported by the Quantum Science Center (QSC), a National Quantum Information Science Research Center of the U.S. Department of Energy (DOE). A.S. and M.C. also acknowledge initial support from the Center for Nonlinear Studies at Los Alamos National Laboratory (LANL). A.S. is now supported by the internal R\&D from Aliro Technologies, Inc. J.L.B. was initially supported by the U.S. DOE through a quantum computing program sponsored by the LANL Information Science \& Technology Institute. J.L.B. was also supported by the National Science Foundation Graduate Research Fellowship under Grant No. 1650115. P.J.C. also acknowledges initial support from the LANL ASC Beyond Moore's Law project.

\bibliography{ref.bib}

\onecolumngrid

\subsection*{Appendix}

\appendix

\section{Proof of Lemma~\ref{lem:bound}}
The fact that the TQFI is a lower bound on the QFI follows directly from the fact that the generalized fidelity is an upper bound for the fidelity. Recall the definition of the QFI and the TQFI, which are respectively defined as
\begin{align}
    I(\theta;\rho_{\theta})=8\lim_{\delta\to0}\frac{1-F(\rho_{\theta},\rho_{\theta+\delta})}{\delta^2}\,,\quad\quad  \mathcal{I}_*(\theta;\rho^{(m)}_{\theta})=8\lim_{\delta\to0}\frac{1-F_*(\rho^{(m)}_{\theta},\rho^{(m)}_{\theta+\delta})}{\delta^2}\,.
\end{align}
From the fact that
    \begin{equation}
        F(\rho_{\theta},\rho_{\theta+\delta})\leq F_*(\rho^{(m)}_{\theta},\rho^{(m)}_{\theta+\delta})\,,
    \end{equation}
we obtain the bound 
\begin{equation}
      \mathcal{I}_*(\theta;\rho^{(m)}_{\theta})\leq   I(\theta;\rho_{\theta})\,.
\end{equation}

Then, let us recall that the generalized fidelity is monotonically  decreasing with $m$~\cite{cerezo2020variationalfidelity}, meaning that we have  $F_*(\rho^{(m)}_{\theta},\rho^{(m)}_{\theta+\delta})\geq F_*(\rho^{(m+1)}_{\theta},\rho^{(m+1)}_{\theta+\delta})$. Hence, from the definition of the TQFI, we find that 
\begin{equation}
      \mathcal{I}_*(\theta;\rho^{(m)}_{\theta})\leq    \mathcal{I}_*(\theta;\rho^{(m+1)}_{\theta})\,.
\end{equation}

\section{Proof of Lemma~\ref{lemma:TQFIexpression}}
\label{app:derivation}
Let us consider a normalized quantum state $\rho_{\theta}=W(\theta)\rho W\ad(\theta)$, where the state $\rho$ has spectral decomposition $\rho = \sum_{j=1}^{r}\lambda_{j}\dya{\lambda_{j}}$, and where $\{\lambda_{j}\}_{j=1}^{r}$ are $\theta$-independent. Then, we have
\begin{align}
    \rho_{\theta}&=\sum_{j=1}^{d}\lambda_{j}e^{-i\theta G}\dya{\lambda_{j}}e^{+i\theta G}=\sum_{j=1}^{d}\lambda_{j}\dya{\lambda_{j}(\theta)}\,,\\
    \rho_{\theta+\delta}&=\sum_{j=1}^{d}\lambda_{j}e^{-i(\theta+\delta) G}\dya{\lambda_{j}}e^{+i(\theta+\delta) G}=\sum_{j=1}^{d}\lambda_{j}e^{-i\delta G}\dya{\lambda_{j}(\theta)}e^{+i\delta G}\,,
\end{align}
where we use the notation 
\begin{align}
    \ket{\lambda_j(\theta)}=e^{-i\theta G}\ket{\lambda_j}\,.
\end{align}
From Eq.~\eqref{eq:TQFIandGeneralizedFidelity}, the TQFI is 
\begin{align}
   \mathcal{I}_*(\theta;\rho^{(m)}_{\theta})=8\lim_{\delta\to0}\frac{1-F_*(\rho^{(m)}_{\theta},\rho^{(m)}_{\theta+\delta})}{\delta^2}\,,
\end{align}
where
\begin{align}\label{eq:Fstar}
    F_{*}(\rho^{(m)}_{\theta},\rho^{(m)}_{\theta+\delta})=\norm{\sqrt{\rho^{(m)}_{\theta}}\sqrt{\rho^{(m)}_{\theta+\delta}}}_{1}+\sqrt{(1-\Tr[\rho^{(m)}_{\theta}])(1-\Tr[\rho^{(m)}_{\theta+\delta}])}\,.
\end{align}
Here, following Ref.~\cite{cerezo2020variationalfidelity}, we can write
\begin{align}
    \norm{\sqrt{\rho_\theta^{(m)}}\sqrt{\rho_{\theta+\delta}^{(m)}}}_1=\Tr\left[\sqrt{T}\right]\,,
\end{align}
where $T$ is an $m\times m$ positive semidefinite operator defined as $T=\sum_{i,j=1}^m T_{ij}\ketbra{\lambda_i}{\lambda_j}$, and where
\begin{align}
    T_{ij}=\sqrt{\lambda_{i}\lambda_{j}}\bra{\lambda_{i}(\theta)}\rho_{\theta+\delta}\ket{\lambda_{j}(\theta)}\,.
\end{align}

For simplicity of notation let us define
\begin{align}
    \IC_{*}(\theta;\rho^{(m)}_{\theta},\delta) = 8\frac{1-F_*(\rho^{(m)}_{\theta},\rho^{(m)}_{\theta+\delta})}{\delta^2}\,,
\end{align}
such that $\IC_{*}(\theta;\rho^{(m)}_{\theta})=\lim_{\delta\to0}\IC_{*}(\theta;\rho^{(m)}_{\theta},\delta)$. To second order in $\delta$, we find 
\begin{align}
    T_{ij}
    =&\sqrt{\lambda_i\lambda_j}\bramatket{\lambda_i}{e^{-i\delta G}\rho e^{i\delta G}   }{\lambda_j}\\
    =&\sqrt{\lambda_i\lambda_j}\bramatket{\lambda_i}{\left(\id -i\delta G-\frac{\delta^2}{2}G^2+\cdots\right) \rho     \left(\id +i\delta G-\frac{\delta^2}{2}G^2+\cdots\right) }{\lambda_j}\\
    =&\sqrt{\lambda_i\lambda_j}~\lambda_i~\delta_{ij}-i\delta\sqrt{\lambda_i\lambda_j}\bramatket{\lambda_i}{\left(G \rho -\rho G \right)}{\lambda_j}+\delta^2\sqrt{\lambda_i\lambda_j}\Bigg(\bramatket{\lambda_i}{G \rho G}{\lambda_j}-\frac{1}{2}\bramatket{\lambda_i}{\left(    G^2 \rho+\rho G^2 \right)}{\lambda_j}\Bigg)+O(\delta^3)\\
  =&\sqrt{\lambda_i\lambda_j}~\lambda_i~\delta_{ij}+i\delta\sqrt{\lambda_i\lambda_j}(\lambda_i-\lambda_j)\bramatket{\lambda_i}{G}{\lambda_j}+\delta^2\sqrt{\lambda_i\lambda_j}\Bigg(\bramatket{\lambda_i}{G \rho G   }{\lambda_j}-\frac{1}{2}(\lambda_i+\lambda_j)\bramatket{\lambda_i}{G^2}{\lambda_j}\Bigg)+O(\delta^3)\,.\label{eq:expansionsecond}
\end{align}
Since we want to find the square root of the operator $T$, we can solve this problem via perturbation by determining an operator $X$ such that $X^2=T$ and $X=\sum_{i,j=1}^m X_{ij}\ketbra{\lambda_i}{\lambda_j}$, with $X$ an $m\times m$ matrix. Hence, from the expansion
\begin{equation}
    X=\sum_{k=0}^\infty\delta^k X^{(k)}\,,
\end{equation}
we find
\begin{equation}\label{eq:perturbation}
    X^2=\sum_{k=0}^{\infty}\delta^k \sum_{p=0}^k X^{(p)}X^{(k-p)}\,.
\end{equation}
To the second order of $\delta$, we can use~\eqref{eq:expansionsecond} to find
\begin{align}
    (X^{(0)})_{ij}&=\left(\sqrt{\lambda_i\lambda_j}\lambda_i\right)^{1/2}\delta_{ij}\,,\\
    (X^{(0)}X^{(1)}+X^{(1)}X^{(0)})_{ij}&=i\sqrt{\lambda_i\lambda_j}(\lambda_i-\lambda_j)\bramatket{\lambda_i}{G}{\lambda_j}\,,\\
     (X^{(0)}X^{(2)}+X^{(2)}X^{(0)}+X^{(1)}X^{(1)})_{ij}&=\sqrt{\lambda_i\lambda_j}\left(\bramatket{\lambda_i}{G\rho G  }{\lambda_j}-\frac{1}{2}(\lambda_i+\lambda_j)\bramatket{\lambda_i}{G^2}{\lambda_j}\right)\,.
\end{align}
These equations allows us to show that 
\begin{align}
    (X^{(1)})_{ij}=&\frac{i\sqrt{\lambda_i\lambda_j}(\lambda_i-\lambda_j)}{\lambda_i+\lambda_j}\bramatket{\lambda_i}{G}{\lambda_j}\,,\\
    (X^{(2)})_{ij}=&\frac{\sqrt{\lambda_i\lambda_j}}{\lambda_i+\lambda_j}\Bigg(\bramatket{\lambda_i}{G\rho G  }{\lambda_j}-\frac{1}{2}(\lambda_i+\lambda_j)\bramatket{\lambda_i}{G^2}{\lambda_j}+\sum_{\ell=1}^{m}\frac{\lambda_{\ell}(\lambda_i-\lambda_{\ell})(\lambda_{\ell}-\lambda_j)}{(\lambda_{i}+\lambda_{\ell})(\lambda_{\ell}+\lambda_{j})}\bramatket{\lambda_i}{G}{\lambda_{\ell}}\bramatket{\lambda_{\ell}}{G}{\lambda_j}\Bigg)\,.
\end{align}
Then, we can compute the trace of $X$ to second order in $\delta$ as
\begin{align}
\begin{split}
\label{eq:expX}
    \norm{&\sqrt{\rho_\theta^{(m)}}\sqrt{\rho_{\theta+\delta}^{(m)}}}_1\\
    &=\Tr[X]=\Tr[X^{(0)}]+\delta\Tr[X^{(1)}]+\delta^2\Tr[X^{(2)}]+O(\delta^3)\\
    &=\sum_{i=1}^m \lambda_i+\frac{\delta^2}{2}\sum_{i=1}^m\left(-\lambda_i\bramatket{\lambda_i}{    G^2 }{\lambda_i}+\sum_{j=1}^d\lambda_{j} |\bramatket{\lambda_i}{    G }{\lambda_{j} }|^2-\sum_{j=1}^m \frac{\lambda_{j} (\lambda_i-\lambda_{j} )^2}{(\lambda_i+\lambda_{j} )^2}|\bramatket{\lambda_i}{    G }{\lambda_{j} }|^2\right)+O(\delta^3)\\
    &=\sum_{i=1}^m \lambda_i-\frac{\delta^2}{2}\sum_{i=1}^m\lambda_i\bramatket{\lambda_i}{G^2}{\lambda_i} +\frac{\delta^2}{2}\sum_{i,j=1}^m\frac{4\lambda_{j} ^2\lambda_i}{(\lambda_i+\lambda_{j} )^2}|\bramatket{\lambda_i}{G}{\lambda_{j} }|^2+\frac{\delta^2}{2}\sum_{i=1}^{m}\sum_{j=m+1}^{d}\lambda_{j} |\bramatket{\lambda_i}{G}{\lambda_{j} }|^2+O(\delta^3)\,.
    \end{split}
\end{align}
Note that throughout our derivations, we use the fact that $\lambda_{j} =0$ for $j>r$. Let us now consider the second term in Eq.~\eqref{eq:Fstar}. To second order in $\delta$ we simply find 
\begin{equation}\label{eq:exptraces}
    \sqrt{\!\left(1\!-\!\Tr\left[\rho^{(m)}_{\theta}\right]\right)\!\!\left(1\!-\!\Tr\left[\rho_{\theta+\delta}^{(m)}\right]\right)}=1-\sum_{i=1}^{m}\lambda_i-\frac{\delta^2}{2}\sum_{i=1}^{m}\left(\sum_{j=1}^{d}\lambda_j|\langle\lambda_i|G|\lambda_j\rangle|^2-\lambda_i\langle\lambda_i|G^2|\lambda_i\rangle\right)+O(\delta^3)
\end{equation}
Then, combining Eqs.~\eqref{eq:expX} and~\eqref{eq:exptraces}, we can obtain
\begin{align}
\begin{split}
  \IC_{*}(\theta,\rho_{\theta}^{(m)})&=\lim_{\delta\to0}\IC_{*}(\theta,\rho_{\theta}^{(m)},\delta)\\
  &=\lim_{\delta\to0}\frac{8}{\delta^2}\left(1-F_{*}(\rho_{\theta}^{(m)},\rho_{\theta+\delta}^{(m)})\right)\\
  &=\lim_{\delta\to0}\frac{8}{\delta^2}\left(1-\norm{\sqrt{\rho_\theta^{(m)}}\sqrt{\rho_{\theta+\delta}^{(m)}}}_1-\sqrt{\!\left(1\!-\!\Tr\left[\rho^{(m)}_{\theta}\right]\right)\!\!\left(1\!-\!\Tr\left[\rho_{\theta+\delta}^{(m)}\right]\right)}\right)\\
  &=\lim_{\delta\to0}\frac{8}{\delta^2}\Bigg(1-\sum_{i=1}^m \lambda_i+\frac{\delta^2}{2}\sum_{i=1}^m\lambda_i\bramatket{\lambda_i}{G^2}{\lambda_i} -\frac{\delta^2}{2}\sum_{i,j=1}^m\frac{4\lambda_{j} ^2\lambda_i}{(\lambda_i+\lambda_{j} )^2}|\bramatket{\lambda_i}{G}{\lambda_{j} }|^2-\frac{\delta^2}{2}\sum_{i=1}^{m}\sum_{j=m+1}^{d}\lambda_{j} |\bramatket{\lambda_i}{G}{\lambda_{j} }|^2\\
  &\quad\quad\quad\quad\quad\quad\quad
  -1+\sum_{i=1}^{m}\lambda_i+\frac{\delta^2}{2}\sum_{i=1}^{m}\sum_{j=1}^{d}\lambda_j|\langle\lambda_i|G|\lambda_j\rangle|^2-\frac{\delta^2}{2}\sum_{i=1}^{m}\lambda_i\bramatket{\lambda_i}{G^2}{\lambda_i}+O(\delta^3)\Bigg)\\
  &=4\sum_{i,j=1}^{m}\lambda_j\abs{\bramatket{\lambda_i}{G}{\lambda_j}}^2-16\sum_{i,j=1}^{m}\frac{\lambda_j^2\lambda_i}{(\lambda_i+\lambda_j)^2}\abs{\bramatket{\lambda_i}{G}{\lambda_j}}^2\\
  &=4\sum_{i,j=1}^{m}\lambda_j\abs{\bramatket{\lambda_i}{G}{\lambda_j}}^2-8\sum_{i,j=1}^{m}\frac{\lambda_i\lambda_j}{\lambda_i+\lambda_j}\abs{\bramatket{\lambda_i}{G}{\lambda_j}}^2\,.
\label{eq:TGQFIlim}
\end{split}
\end{align}
Here, in the last equality we used the fact that 
\begin{align}
\begin{split}
    2\sum_{i,j=1}^{m}\frac{\lambda_i\lambda_{j} ^2}{(\lambda_i+\lambda_{j} )^2}|\bramatket{\lambda_i}{G}{\lambda_{j} }|^2&=\sum_{i,j=1}^{m}\frac{\lambda_i\lambda_{j} ^2}{(\lambda_i+\lambda_{j} )^2}|\bramatket{\lambda_i}{G}{\lambda_{j} }|^2+\sum_{i,j=1}^{m}\frac{\lambda_i^2\lambda_{j} }{(\lambda_i+\lambda_{j} )^2}|\bramatket{\lambda_i}{G}{\lambda_{j} }|^2\\
    &=\sum_{i,j=1}^{m}\frac{\lambda_i\lambda_{j} }{\lambda_i+\lambda_{j} }|\bramatket{\lambda_i}{G}{\lambda_{j} }|^2
\end{split}
\end{align}
Also, we remark that this can be also simplified as
\begin{align}
\IC_{*}(\theta,\rho_{\theta}^{(m)})&=2\sum_{i,j=1}^m \frac{(\lambda_i-\lambda_{j} )^2}{\lambda_i+\lambda_{j} }|\bramatket{\lambda_i}{G}{\lambda_{j} }|^2\,,
\end{align}
and finally, because of the symmetry of the summand in $i$ and $j$, we have
\begin{align}
\IC_{*}(\theta,\rho_{\theta}^{(m)})&=4\sum_{i<j}^m \frac{(\lambda_i-\lambda_{j} )^2}{\lambda_i+\lambda_{j} }|\bramatket{\lambda_i}{G}{\lambda_{j} }|^2\,,
\end{align}

\section{Truncated Symmetric Logarithmic Derivative}
The standard QFI can be defined in terms of the so-called SLD operator. For the state $\rho_{\theta}=\sum_{i=1}^d \lambda_i \ket{\lambda_i (\theta)}\bra{\lambda_i (\theta)}$, the standard SLD operator is~\cite{Mateo09}
\begin{align}
    J_{\theta} = 2\sum_{i,j=1}^d \frac{\bra{\lambda_i (\theta)}\partial_{\theta}\rho_{\theta}\ket{\lambda_j (\theta)}}{\lambda_i +\lambda_j}\ket{\lambda_i (\theta)}\bra{\lambda_j (\theta)}.
\end{align}
Analogously, we can also define the TQFI through a truncated SLD (TSLD) operator. Let the spectral decomposition of our truncated exact state be given as
\begin{align}
    \tau_{\theta} =\sum_{i=1}^m \lambda_i \ket{\lambda_i (\theta)}\bra{\lambda_i(\theta)},
\end{align}
which the parameter dependence left implicit to simplify notation. Then, the TSLD operator is
\begin{align}
     L_{\theta} = 2\sum_{i,j=1}^m \frac{\bra{\lambda_i (\theta)}\partial_{\theta}\tau_{\theta}\ket{\lambda_j (\theta)}}{\lambda_i +\lambda_j}\ket{\lambda_i (\theta)}\bra{\lambda_j (\theta)}.
\end{align}
One can easily verify that $\partial_{\theta} \tau_{\theta} = i [\tau_{\theta},G]$, so that the explicit form of the TSLD operator becomes
\begin{align}
\begin{split}
     L_{\theta} &= 2i\sum_{i,j=1}^m \frac{\lambda_i - \lambda_j}{\lambda_i +\lambda_j}\bra{\lambda_i (\theta)}G\ket{\lambda_j (\theta)}\ket{\lambda_i (\theta)}\bra{\lambda_j (\theta)}\\
     &=2i\sum_{i,j=1}^m \frac{\lambda_i - \lambda_j}{\lambda_i +\lambda_j}\bra{\lambda_i}G\ket{\lambda_j }W(\theta)\ket{\lambda_i}\bra{\lambda_j}W\ad(\theta),
\end{split}
\label{eq:TSLDop}
\end{align}
where we used the fact that $\ket{\lambda_j(\theta)}=W(\theta)\ket{\lambda_j}$ and $[W(\theta),G]=0$.
Taking the conjugate transpose and then exchanging $i$ and $j$, we can easily verify that $L_{\theta}$ is Hermitian, i.e., $L_{\theta}=L_{\theta}\ad$.  In addition to Hermiticity, the justification for regarding $L_{\theta}$ as an SLD operator comes from the following propositions.  
\begin{proposition}
For a $\theta$-parametrized subnormalized state $\tau_{\theta}=W(\theta)\tau W\ad(\theta)$, the TSLD operator satisfies
    \begin{align}
        \partial_{\theta}\tau_{\theta} = \frac{1}{2}(  L_{\theta}  \tau_{\theta} + \tau_{\theta}  L_{\theta} )\,,
    \end{align}
with $\Tr[ L_{\theta} \tau_{\theta}] = 0\,.$
\end{proposition}    

\begin{proof}
    First, using the fact that $\partial_{\theta} \tau_{\theta} = i[\tau_{\theta},G]$ and $\tau_{\theta}\ket{\lambda_i(\theta)}=\lambda_i\ket{\lambda_i(\theta)}$, we can write
    \begin{align}
        \partial_{\theta} \tau_{\theta} &= i \sum_{i,j=1}^m \bra{\lambda_i(\theta)}[\tau_{\theta},G]\ket{\lambda_j(\theta)}\ket{\lambda_i(\theta)}\bra{\lambda_j(\theta)},\\
        &=i \sum_{i,j=1}^m (\lambda_i -\lambda_j)\bra{\lambda_i}G\ket{\lambda_j}\ket{\lambda_i(\theta)}\bra{\lambda_j(\theta)}.
    \end{align}
    Then, making use of the explicit expansions of $\tau_{\theta}$ and $L_{\theta}$ in the eigenbasis of $\tau_{\theta}$, we can write
    \begin{align}
    \frac{1}{2} ( L_{\theta}  \tau_{\theta}+\tau_{\theta} L_{\theta}  )    
    &=i\sum_{i,j=1}^{m}\left(\frac{\lambda_{i}-\lambda_{j}}{\lambda_{i}+\lambda_{j}}\right)\bra{\lambda_{i}(\theta)}G\ket{\lambda_j(\theta)}(\ket{\lambda_{i}(\theta)}\bra{\lambda_{j}(\theta)}\tau_{\theta}+\tau_{\theta} \ket{\lambda_{i}(\theta)}\bra{\lambda_{j}(\theta)})\\
    &=i\sum_{i,j=1}^{m}\left(\frac{\lambda_{i}-\lambda_{j}}{\lambda_{i}+\lambda_{j}}\right)\bra{\lambda_{i}}G\ket{\lambda_j}(\lambda_i+\lambda_j)\ket{\lambda_{i}(\theta)}\bra{\lambda_{j}(\theta)}\\
    &=i \sum_{i,j=1}^m (\lambda_i -\lambda_j)\bra{\lambda_i}G\ket{\lambda_j}\ket{\lambda_i(\theta)}\bra{\lambda_j(\theta)}\,.
    \end{align}
    Comparing these two expressions, we see that indeed 
    \begin{align}
        \partial_{\theta}\tau_{\theta} = \frac{1}{2}(L_{\theta}\tau_{\theta}+\tau_{\theta} L_{\theta}),
    \end{align}
    as is required of a well-defined SLD operator. Finally, because $\Tr[\tau_{\theta},G]=0$, we have
    \begin{align}
        \Tr[\partial_{\theta}\tau_{\theta}] = i\Tr[\tau_{\theta},G]= \Tr\left[\frac{1}{2}(L_{\theta}\tau_{\theta}+\tau_{\theta} L_{\theta})\right]=\Tr[L_{\theta}\tau_{\theta}]=0\,.
    \end{align}
\end{proof}

\begin{proposition}
The TQFI $\mathcal{I}_{*}(\theta;\tau_{\theta})$  can be expressed as 
\begin{align}
    \mathcal{I}_{*}(\theta;\tau_{\theta})= \Tr[L_{\theta}^2\tau_{\theta}]\,.
\end{align}
\end{proposition}
    
\begin{proof}
    $\mathcal{I}_{*}(\theta;\tau_{\theta})$ is given by 
    \begin{align}
    \mathcal{I}_*(\theta;\tau_{\theta}) =
    4\sum_{i,j=1}^m \lambda_i|\bra{\lambda_{i}}G\ket{\lambda_{j}}|^2-8\sum_{i,j=1}^m\frac{\lambda_i\lambda_j}{\lambda_i+\lambda_j}|\bra{\lambda_{i}}G\ket{\lambda_{j}}|^2
    \label{eq:TQFIappendix}
    \end{align}
    From Eq.~\eqref{eq:TSLDop}, we have
    \begin{align}
        L_{\theta} = 2i\sum_{i,j=1}^{m}\frac{\lambda_i-\lambda_j}{\lambda_i+\lambda_j}\bra{\lambda_i}G\ket{\lambda_j}W(\theta)\ket{\lambda_i}\bra{\lambda_j}W\ad(\theta)\,,
    \end{align}
    and
    \begin{align}
    \tau_{\theta}=W(\theta)\tau W\ad(\theta)=\sum_{i=1}^{m}\lambda_i W(\theta)\dya{\lambda_i} W\ad(\theta)\,.
    \end{align}
    Therefore, we obtain
    \begin{align}
        \Tr[ L_{\theta}^2 \tau_{\theta}] =&-4\sum_{\substack{i,j=1\\k,\ell=1}}^{m}\sum_{r=1}^{m}\left(\frac{\lambda_{i}-\lambda_{j}}{\lambda_{i}+\lambda_{j}}\right)\left(\frac{\lambda_{k}-\lambda_{\ell} }{\lambda_{k}+\lambda_{\ell} }\right)\bra{\lambda_{i}}G\ket{\lambda_{j}}\bra{\lambda_{k}}G\ket{\lambda_{\ell} }\cdot\lambda_{r}\cdot\delta_{jk}\delta_{\ell r}\delta_{ir}\\
        =&4\sum_{i,j=1}^{m}\lambda_{i}\left(\frac{\lambda_{i}-\lambda_{j}}{\lambda_{i}+\lambda_{j}}\right)^2\abs{\bra{\lambda_{i}}G\ket{\lambda_{j}}}^2\\
        =&4\sum_{i,j=1}^{m}\frac{\lambda_{i}^3-2\lambda_{i}^2\lambda_{j}+\lambda_{i}\lambda_{j}^2}{(\lambda_{i}+\lambda_{j})^2}\abs{\bra{\lambda_{i}}G\ket{\lambda_{j}}}^2\\
        =&4\sum_{i,j=1}^{m}\frac{\lambda_{i}^3+\lambda_{i}\lambda_{j}^2}{(\lambda_{i}+\lambda_{j})^2}\abs{\bra{\lambda_{i}}G\ket{\lambda_{j}}}^2-4\sum_{i,j=1}^{m}\frac{2\lambda_{i}^2\lambda_{j}}{(\lambda_{i}+\lambda_{j})^2}\abs{\bra{\lambda_{i}}G\ket{\lambda_{j}}}^2\\
        =&2\sum_{i,j=1}^{m}\frac{\lambda_{i}^3+\lambda_{i}\lambda_{j}^2+ \lambda_{j}^3+\lambda_{j}\lambda_{i}^2}{(\lambda_{i}+\lambda_{j})^2}\abs{\bra{\lambda_{i}}G\ket{\lambda_{j}}}^2-4\sum_{i,j=1}^{m}\frac{\lambda_{i}^2\lambda_{j}+\lambda_{j}^2\lambda_{i}}{(\lambda_{i}+\lambda_{j})^2}\abs{\bra{\lambda_{i}}G\ket{\lambda_{j}}}^2\\
        =&2\sum_{i,j=1}^{m}\frac{(\lambda_{i}+\lambda_{j})^3-2\lambda_{i}\lambda_{j}(\lambda_{i}+\lambda_{j})}{(\lambda_{i}+\lambda_{j})^2}\abs{\bra{\lambda_{i}}G\ket{\lambda_{j}}}^2-4\sum_{i,j=1}^{m}\frac{\lambda_{i}\lambda_{j}(\lambda_{i}+\lambda_{j})}{(\lambda_{i}+\lambda_{j})^2}\abs{\bra{\lambda_{i}}G\ket{\lambda_{j}}}^2\\
        =&2\sum_{i,j=1}^{m}(\lambda_{i}+\lambda_{j})\abs{\bra{\lambda_{i}}G\ket{\lambda_{j}}}^2-8\sum_{i,j=1}^{m}\frac{\lambda_{i}\lambda_{j}}{\lambda_{i}+\lambda_{j}}\abs{\bra{\lambda_{i}}G\ket{\lambda_{j}}}^2\\
        =&4\sum_{i,j=1}^{m}\lambda_{i}\abs{\bra{\lambda_{i}}G\ket{\lambda_{j}}}^2-8\sum_{i,j=1}^{m}\frac{\lambda_{i}\lambda_{j}}{\lambda_{i}+\lambda_{j}}\abs{\bra{\lambda_{i}}G\ket{\lambda_{j}}}^2\,,
    \end{align}
    which leads to
    \begin{align}
       \mathcal{I}_{*}(\theta;\tau_{\theta})=\Tr[ L_{\theta}^2 \tau_{\theta}]\,.
    \end{align}
\end{proof}

\section{Proof of Lemma~\ref{lemma:TQFIproperty}}
\label{app:property}

In the following, we prove each property in Lemma~\ref{lemma:TQFIproperty}.
\begin{enumerate}
    \item{Invariance under unitary transformations:
    Given a unitary $V\in U(d)$ which is $\theta$-independent,  since the generalized fidelity~\cite{tomamichel2015quantum} is unitary-invariant, i.e., 
    \begin{align}
        F_{*} (V\tau_{\theta}V\ad,V\tau_{\theta+\delta}V\ad  ) = F_{*} (\tau_{\theta},\tau_{\theta+\delta}  )\,,
    \end{align}
    we obtain
    \begin{align}
        \mathcal{I}_{*}(\theta;V\tau_{\theta}V\ad  )=\mathcal{I}_{*}(\theta;\tau_{\theta}  )\,.
    \end{align}
    }
     \item{Convexity:
     Let $\tau_{\theta}$ and $\xi_{\theta}$ be subnormalized states. Since the generalized fidelity is jointly concave~\cite{tomamichel2015quantum}, we have
    \begin{align}
    F_{*}\Big(q\tau_{\theta}+(1-q)\xi_{\theta},q\tau_{\theta+\delta}+(1-q)\xi_{\theta+\delta}\Big)\geq q F_{*} (\tau_{\theta},\tau_{\theta+\delta}  )+(1-q) F_{*} (\xi_{\theta},\xi_{\theta+\delta}  )\,.
    \end{align}
    Hence, for all $\delta$, we have
    \begin{align}
         \mathcal{I}_{*}(\theta;q\tau_{\theta}+(1-q)\xi_{\theta})\leq q \mathcal{I}_{*}(\theta;\tau_{\theta}  )+(1-q)\mathcal{I}_{*}(\theta;\xi_{\theta})\,.
    \end{align}
    }
    \item{Monotonicity under CPTNI map:
    Here, we employ the monotonicity of the generalized fidelity~\cite{tomamichel2015quantum}: 
    \begin{align}
        F_{*} (\tau_{\theta},\tau_{\theta+\delta}  )\leq F_{*} (\Phi (\tau_{\theta}  ),\Phi (\tau_{\theta+\delta}  )  )
    \end{align}
    for a CPTNI map $\Phi$. 
    For all $\delta$, we have
    \begin{align}
        \frac{1-F_{*}(\tau_{\theta},\tau_{\theta+\delta})}{\delta^2} \geq \frac{1-F_{*}(\Phi(\tau_{\theta}),\Phi(\tau_{\theta+\delta}))}{\delta^2}\,,
    \end{align}
    so that  we have  
    \begin{align}
        \mathcal{I}_{*}(\theta;\tau_{\theta}  )\geq \mathcal{I}_{*}(\theta;\Phi(\tau_{\theta}))\,.
    \end{align}
    }
   
    \item{Subadditivity  for  product  of  truncated  states:
    Consider a subnormalized state obtained from a tensor product of subnormalized state $\tau_{\theta}=\bigotimes_{k}\tau_{\theta}^{(k)}$. We have
    \begin{align}
        \partial_{\theta}\tau_{\theta} = \sum_{k}\partial_{\theta}\tau_{\theta}^{(k)}\otimes \tau_{\theta}^{\overline{(k)}}
        =\sum_{k}\frac{L_{k,\theta}\tau_{\theta}^{(k)}+\tau_{\theta}^{(k)}L_{k,\theta}}{2}\otimes \tau_{\theta}^{\overline{(k)}}\,,
    \end{align}
    where we define
    \begin{align}
        \tau_{\theta}^{\overline{(k)}} =\bigotimes_{j\neq k} \tau_{\theta}^{(j)}\,.
    \end{align}
    Therefore, the TSLD operator becomes
    \begin{align}
         L_{\theta} =\sum_{k}L_{k,\theta}\otimes\openone_{\overline{k}}\,.
    \end{align}
    Then, we can obtain
    \begin{align}
        \mathcal{I}_{*}(\theta;\tau_{\theta})=\Tr [L_{\theta}^2\tau_{\theta}] = \sum_{k}\mathcal{A}_{k}\Tr [L_{k,\theta}^2\tau_{\theta}^{(k)}]\,,
    \end{align}
    where
    \begin{align}
        \mathcal{A}_{k}=\prod_{j\neq k}\Tr\left[\tau^{(j)}_{\theta}\right]\leq 1\,.
    \end{align}
    Therefore, 
    \begin{align}
        \mathcal{I}_{*}(\theta;\tau_{\theta})=\sum_{k}\mathcal{A}_{k}\mathcal{I}_*(\theta;\tau_{\theta}^{(k)})\leq \sum_{k}\mathcal{I}_*(\theta;\tau_{\theta}^{(k)})\,.
  \end{align}
  
    }
    \item{Additivity for direct sum of truncated states:
    For $\tau_{\theta}=\bigoplus_{k}\mu_k\tau_{\theta}^{(k)}$, where $\mu_k$ is $\theta$-independent and $0<\sum_{k}\mu_{k}\leq 1$,  we have
    \begin{align}
        \partial_{\theta}\tau_{\theta}=\bigoplus_{k}\mu_k \partial_{\theta}\tau_{\theta}^{(k)}
        =\bigoplus_{k}\mu_k \frac{L_{k,\theta}\tau_{\theta}^{(k)}+\tau_{\theta}^{(k)}L_{k,\theta}}{2}\,.
    \end{align}
    Then, the TSLD operator becomes
    \begin{align}
         L_{\theta} =\bigoplus_{k}L_{k,\theta}\,.
    \end{align}
    Therefore, 
    \begin{align}
        \mathcal{I}_{*}(\theta;\tau_{\theta})=\Tr [ L_{\theta}^2 \tau_{\theta} ]=\sum_{k}\mu_{k}\Tr [L_{k,\theta}^2\tau_{\theta}^{(k)}]\,,
    \end{align}
    so that 
    \begin{align}
        \mathcal{I}_{*}(\theta;\tau_{\theta})=\sum_{k}\mu_k \mathcal{I}_{*}(\theta;\tau_{\theta}^{(k)})\,,
    \end{align}
    }
\end{enumerate}
and $\sqrt{T}$ is taken to be the unique, positive semi-definite square root of $T$. 

Finally, let us remark that an alternative proof for the subadditivity  for  the product  of  truncated  states can be obtained as follows. {This is for the readers who are not familiar with the SLD operator.} It will suffice to show it is true in the bipartite case (as larger product states follow by induction). So, we consider a subnormalized state of the form
\begin{align}
    \tau_{\theta} = \tau_{\theta}^{(1)} \otimes \tau_{\theta}^{(2)}\in \SC_{\leq}(\HC_1)\otimes\SC_{\leq}(\HC_2)\,,
\end{align}
where $\dim(\HC_k)=d_k$ with $k=1,2$. Here, we focus on the unitary families. First, recall that the TQFI is defined as
\begin{align}
\IC_{*}(\theta,\rho_{\theta}^{(m)})&=2\sum_{i,j=1}^m \frac{(\lambda_i-\lambda_{j} )^2}{\lambda_i+\lambda_{j} }|\bramatket{\lambda_i}{G}{\lambda_{j} }|^2\,.
\end{align}
As this explicit form depends on the eigensystem of our state and the generator of the unitary dynamics, let us prove explicitly, when the subnormalized state of the subspace belongs to the unitary families, we have
\begin{align}
    \tau_{\theta} &= \tau_{\theta}^{(1)} \otimes \tau_{\theta}^{(2)} \\
    &= e^{-i\theta G^{(1)}} \tau^{(1)} e^{+i\theta G^{(1)}} \otimes e^{-i\theta G^{(2)}} \tau^{(2)} e^{+i\theta G^{(2)}}\\
    &= (e^{-i\theta G^{(1)}} \otimes e^{-i\theta G^{(2)}})(\tau^{(1)}\otimes \tau^{(2)})(e^{+i\theta G^{(1)}} \otimes e^{+i\theta G^{(2)}})\,.
\end{align}
A useful quantity needed here is the Kronecker sum defined as
\begin{align}
    A^{(1)} \oplus B^{(2)} = A^{(1)}\otimes \openone^{(2)} + \openone^{(1)} \otimes B^{(2)}\,,
\end{align}
and we recall here the following useful identity  
\begin{align}
    e^{A} \otimes e^{B} = e^{A\oplus B}=e^{A^{(1)}\otimes \openone^{(2)} + \openone^{(1)} \otimes B^{(2)}}\,,
\end{align}
where $\oplus$ is the Kronecker sum defined above. Hence, we have
\begin{align}
    \tau_{\theta} &= (e^{-i\theta G^{(1)}} \otimes e^{-i\theta G^{(2)}})(\tau^{(1)}\otimes \tau^{(2)})(e^{+i\theta G^{(1)}} \otimes e^{+i\theta G^{(2)}}) \\
    &= e^{-i \theta(G^{(1)}\otimes \openone^{(2)} + \openone^{(1)} \otimes G^{(2)})} (\tau^{(1)} \otimes \tau^{(2)}) e^{+i \theta(G^{(1)}\otimes \openone^{(2)} + \openone^{(1)} \otimes G^{(2)})}\,.
\end{align}

As for the eigensystem, we note that the Hilbert space is now of the form $\mathcal{H}^{(1)} \otimes \mathcal{H}^{(2)}$ so the eigenvalues and eigenvectors are now of the form 
\begin{align}
    \tau_{\theta}^{(1)} \otimes \tau_{\theta}^{(2)} \ket{\lambda_{i}} \otimes \ket{\lambda_{j}} = \lambda_{i} \lambda_{j} \ket{\lambda_{i}} \otimes \ket{\lambda_{j}}\,,
\end{align}
where we have that $i \in [1,m_1]$ and $j\in[1,m_2]$, and where $0<\sum_i \lambda_i \leq 1$ and $0<\sum_j \lambda_j \leq 1$. Together, TQFI becomes
\begin{align}
      \IC_{*}(\theta,\tau_{\theta}^{(1)} \otimes\tau_{\theta}^{(2)})    &=2 \sum_{i,k=1}^{m_1} \sum_{j,\ell=1}^{m_2} \frac{(\lambda_{i} \lambda_{j} - \lambda_{k} \lambda_{\ell} )^2}{\lambda_{i} \lambda_{j} +  \lambda_{k} \lambda_{\ell} } |\bra{\lambda_{i}}\otimes \bra{\lambda_{j}} (G^{(1)}\otimes \openone^{(2)} + \openone^{(1)} \otimes G^{(2)}) \ket{\lambda_{k}}\otimes \ket{\lambda_{\ell} }|^2\,.
\end{align}
Let us now expand the matrix element part of the expression
\begin{align}
    &|\bra{\lambda_{i}}\otimes \bra{\lambda_{j}} (G^{(1)}\otimes \openone^{(2)} + \openone^{(1)} \otimes G^{(2)}) \ket{\lambda_{k}}\otimes \ket{\lambda_{\ell} }|^2 \\
    &= (\bra{\lambda_{i}}\otimes \bra{\lambda_{j}} (G^{(1)}\otimes \openone^{(2)} + \openone^{(1)} \otimes G^{(2)}) \ket{\lambda_{k}}\otimes \ket{\lambda_{\ell} })\nonumber\\
    &\times (\bra{\lambda_{k}}\otimes \bra{\lambda_{\ell} } (G^{(1)}\otimes \openone^{(2)} + \openone^{(1)} \otimes G^{(2)}) \ket{\lambda_{i}}\otimes \ket{\lambda_{j}})\\
    &= \big(\bra{\lambda_{i}}G^{(1)}\ket{\lambda_{k}}\delta_{j\ell} + \delta_{ik}\bra{\lambda_{j}}G^{(2)}\ket{\lambda_{\ell} } \big)\big(\bra{\lambda_{k}}G^{(1)}\ket{\lambda_{i}}\delta_{j\ell} + \delta_{ik}\bra{\lambda_{\ell} }G^{(2)}\ket{\lambda_{j}} \big)\\
    &= |\bra{\lambda_{i}}G^{(1)}\ket{\lambda_{k}}|^2 \delta_{j\ell} \delta_{j\ell} +(\bra{\lambda_{i}}G^{(1)}\ket{\lambda_{k}})(\bra{\lambda_{\ell} }G^{(2)}\ket{\lambda_{j}})\delta_{ik}\delta_{j\ell}\nonumber\\
    &+(\bra{\lambda_{\ell}}G^{(2)}\ket{\lambda_{j} })(\bra{\lambda_{k}}G^{(1)}\ket{\lambda_{i}})\delta_{j\ell}\delta_{ik} + |\bra{\lambda_{j}}G^{(2)}\ket{\lambda_{\ell} }|^2 \delta_{ik}\delta_{ik}
\end{align}
Replacing this expansion in the summation of TQFI sum, we see that the terms with  $\delta_{ik}\delta_{j\ell}$ lead to $\lambda_{i} \lambda_{j} - \lambda_{k} \lambda_{\ell} =0$. Hence, the first term in the TQFI becomes
\begin{align}
&2 \sum_{i,k=1}^{m_1} \sum_{j,\ell=1}^{m_2} \frac{(\lambda_{i} \lambda_{j} - \lambda_{k} \lambda_{\ell} )^2}{\lambda_{i} \lambda_{j} +  \lambda_{k} \lambda_{\ell} } |\bra{\lambda_{i}}\otimes \bra{\lambda_{j}} (G^{(1)}\otimes \openone^{(2)} + \openone^{(1)} \otimes G^{(2)}) \ket{\lambda_{k}}\otimes \ket{\lambda_{\ell} }|^2\\
   &=2 \sum_{i,k=1}^{m_1} \sum_{j,\ell=1}^{m_2} \frac{(\lambda_{i} \lambda_{j} - \lambda_{k} \lambda_{\ell} )^2}{\lambda_{i} \lambda_{j} +  \lambda_{k} \lambda_{\ell} } \left(|\bra{\lambda_{i}}G^{(1)}\ket{\lambda_{k}}|^2 \delta_{j\ell}  +|\bra{\lambda_{j}}G^{(2)}\ket{\lambda_{\ell} }|^2 \delta_{ik}\right) \\
     &= 2 \sum_{i,k=1}^{m_1} \sum_{j=1}^{m_2} \frac{\lambda_{j}^2 (\lambda_{i}  - \lambda_{k} )^2}{\lambda_{j}(\lambda_{i}  +  \lambda_{k}) } |\bra{\lambda_{i}}G^{(1)}\ket{\lambda_{k}}|^2 +2  \sum_{i=1}^{m_1} \sum_{j=1}^{m_2} \frac{\lambda_{i}^2 (\lambda_{j}  - \lambda_{\ell} )^2}{\lambda_{i}(\lambda_{j}  +  \lambda_{\ell}) } |\bra{\lambda_{j}}G^{(2)}\ket{\lambda_{\ell}}|^2\\
     &\leq 2 \sum_{i,k=1}^{m_1} \frac{ (\lambda_{i}  - \lambda_{k} )^2}{\lambda_{i}  +  \lambda_{k}} |\bra{\lambda_{i}}G^{(1)}\ket{\lambda_{k}}|^2 + 2 \sum_{j,\ell=1}^{m_2} \frac{(\lambda_{j} - \lambda_{\ell} )^2}{ \lambda_{j} + \lambda_{\ell} } |\bra{\lambda_{j}}G^{(2)}\ket{\lambda_{\ell} }|^2\,, 
\end{align}
Therefore, we get 
\begin{align}
    \IC_{*}(\theta,\tau_{\theta}^{(1)} \otimes\tau_{\theta}^{(2)} ) \leq \IC_{*}(\theta,\tau_{\theta}^{(1)}) + \IC_{*}(\theta,\tau_{\theta}^{(2)})\,
\end{align}
as desired. 

\section{Proof of Lemma~\ref{lemma:Bures}}
\label{app:proofBures}
Let $\sigma,~\xi,$ and $\eta$ be subnormalized states in $\SC_{\leq}(\HC)$. Then we have that the following properties of the generalized Bures distance hold:
\begin{enumerate}
    \item{Symmetry:  Because of $F_{*} (\sigma,\xi  )=F_{*} (\xi,\sigma  )$, we have $B_{*} (\sigma,\xi  ) = B_{*} (\xi,\sigma  )$.}
    \item{Identity of indiscernibles: Because $F_{*} (\sigma,\xi  )=1$ if and only if $\sigma=\xi$, we have $B_{*} (\sigma,\xi  ) =0$ if and only if $\sigma=\xi$.}
    \item{Triangular inequality: Let $A_{*} (\sigma,\xi  )$ be the generalized angular distance $ A_{*} (\sigma,\xi  ) =\arccos (F_{*} (\sigma,\xi  )  ),$
    and $0 \leq A_{*} (\sigma,\xi  )\leq\frac{\pi}{2}$.
    Then, we can write
    \begin{align}
        B_{*} (\sigma,\xi)=2\sin \left(\frac{A_{*} (\sigma,\xi)}{2}\right)\,.
    \end{align}
    From the triangle inequality for the generalized angular distance~\cite{tomamichel2010duality,tomamichel2015quantum}, we have
    \begin{align}
    \begin{split}
        B_{*} (\tau,\xi  )=2\sin \left(\frac{A_{*} (\sigma,\xi)}{2}\right)
        \leq 2\sin \left(\frac{A_{*} (\sigma,\eta  )}{2}  \right)+2\sin \left(\frac{A_{*} (\eta,\xi  )}{2}  \right)=B_{*} (\sigma,\eta  )+B_{*} (\eta,\xi  )\,.
    \end{split}
    \end{align}
    }
\end{enumerate}
These prove that $B_{*}(\sigma,\xi)$ is a distance metric on the space of subnormalized state. Let us finally remark that the generalized Bures distance can also be expressed as  $B_{*}(\sigma,\xi)=2P^2(\sigma,\xi)$, where  $P(\sigma,\xi)=\sqrt{1-F_*(\sigma,\xi)}$ is the so-called purified distance~\cite{tomamichel2010duality,tomamichel2015quantum}.

\section{{Proof of Lemma~\ref{lemma:JandB}}}
\label{app:proofBandI}

Let us consider $B_{*}^{2} (\rho_{\theta}^{(m)},\rho_{\theta+\delta}^{(m)})$. Suppose that $B_{*}^{2} (\rho_{\theta}^{(m)},\rho_{\theta+\delta}^{(m)})$ has the form
\begin{align}
    B_*^{2}(\rho_{\theta}^{(m)},\rho_{\theta+\delta}^{(m)}) = \frac{1}{4}\sum_{k=0}^{\infty}b_k \delta^k = \frac{b_0}{4}+\frac{b_1}{4}\delta+\frac{b_2}{4}\delta^2+O(\delta^3)\,,
\end{align}
where $b_k\in\mathbb{R}$ and $|b_k|<\infty$.
Defining $F_{*}(\delta)=F_{*}(\rho_{\theta}^{(m)},\rho_{\theta+\delta}^{(m)})$,  by definition, we have
\begin{align}
  B_*^{2}(\rho_{\theta}^{(m)},\rho_{\theta+\delta}^{(m)})= 2 (1-F_{*} (\delta)) 
  = 2 -2 \left(F_{*}(0)+\delta \partial_{\delta}F_{*}(0)+\frac{\delta^2}{2}\partial_{\delta}^2 F_{*}(0)+O(\delta^3)\right)\,,
\end{align}
From $F_*(0)=1$ and $\partial_{\delta}F_{*}(0)=0$ because the generalized fidelity is a continuous function of $\delta$ and  becomes maximum at $\delta=0$, we have
\begin{align}
   B_*^{2}(\rho_{\theta}^{(m)},\rho_{\theta+\delta}^{(m)})=-\delta^2\partial_{\delta}^2 F_{*}(0)+O(\delta^3)\,. 
\end{align}
Therefore, we arrive at the following equality
\begin{align}
    \frac{b_0}{4}+\frac{b_1}{4}\delta+\frac{b_2}{4}\delta^2+O (\delta^3) = -\delta^2\partial_{\delta}^2 F_{*}(0)+O(\delta^3)\,,
\end{align}
which has to be valid for any infinitesimal $\delta$. Therefore, we must have $b_0=b_1=0$, and 
\begin{align}
    b_2 = -4\partial_{\delta}^2 F_{*}(0)\,.
\end{align}

Here, applying $F_*(0)=1$ and $\partial_{\delta}F_{*}(0)=0$, by definition of the truncated QFI, we can also obtain   
\begin{align}
    \mathcal{I}_{*}(\theta;\rho_{\theta}^{(m)})= 8\lim_{\delta\to 0}\frac{1-F_{*}(\delta)}{\delta^2} = -4\partial_{\delta}^2 F_{*}(0)\,,
\end{align}
which leads to 
\begin{align}
    b_2 = \mathcal{I}_{*}(\theta;\rho_{\theta}^{(m)})\,.
\end{align}
Therefore, we obtain
\begin{align}
    B_*^{2}(\rho_{\theta}^{(m)},\rho_{\theta+\delta}^{(m)}) = \frac{1}{4}\mathcal{I}_{*}(\theta;\rho_{\theta}^{(m)})\delta^2+O (\delta^3)\,.
\end{align}

\end{document}